\documentclass[12pt]{article}

\usepackage{a4wide, amsmath,amsthm,amsfonts,amscd,amssymb,eucal,bbm}

\def\RR{{\mathbb R}}
\def\CC{{\mathbb C}}
\def\NN{{\mathbb N}}

\def\A{{\mathcal A}}
\def\B{{\mathcal B}}

\def\H{{\mathcal H}}

\def\K{{\mathcal K}}
\def\M{{\mathcal M}}
\def\N{{\mathcal N}}

\def\R{{\mathcal R}}

\def\W{{\mathcal W}}

\newtheorem{theorem}{Theorem}[section]

\newtheorem{corollary}[theorem]{Corollary}
\newtheorem{proposition}[theorem]{Proposition}
\newtheorem{conjecture}[theorem]{Conjecture}
\newtheorem{lemma}[theorem]{Lemma}

\begin{document}
\date{}

\title{An algebraic Haag's theorem}

\author{{\bf Mih\'aly Weiner\footnote{On leave from the Alfr\'ed R\'enyi 
Institute of Mathematics}
\footnote{Supported by the ERC 
Advanced Grant 227458 OACFT ``Operator Algebras and Conformal Field
Theory''}}\\ 
Universit\`a di Roma ``Tor Vergata'', 
Dipartimento di Matematica \\
{\tt{mweiner@renyi.hu}}}

\maketitle

\begin{abstract}
Under natural conditions (such as split property and geometric modular 
action of wedge algebras) it is shown that the unitary equivalence class 
of the net of local (von Neumann) algebras in the vacuum sector
associated to double cones with bases on a fixed space-like hyperplane
completely determines an algebraic QFT model. More precisely, if for two 
models there is unitary connecting all of these algebras, then --- {\it 
without} assuming that this unitary also connects their 
respective vacuum states or spacetime symmetry representations --- it 
follows that the two models are equivalent. This result might be viewed 
as an algebraic version of the celebrated  theorem of Rudolf Haag 
about problems regarding the so-called ``interaction-picture'' in QFT. 

Original motivation of the author for finding such an algebraic version 
came from conformal chiral QFT. Both the chiral case as well as a related 
conjecture about standard half-sided modular inclusions will be also 
discussed.
\end{abstract}

\numberwithin{equation}{section}

\section{Introduction}
\subsection{Haag's theorem and its algebraic version}
 
If we ``freeze'' a classical, nonrelativistic physical system --- say a 
mechanical system of $n$ point masses --- at a certain time-instant, we 
do not see if the system was an ``interactive'' or a ``free'' one. 
A certain configuration with given velocities may correspond both to a free 
or to an interactive system. Interaction becomes visible only when one 
looks at how things {\it change}.

This is the basic idea behind the so-called ``interaction-picture'' 
in quantum field theory (QFT). Within the framework of {\it 
Wightmann-axioms} \cite{andallthat}, free models can be well-described in 
terms of Wightmann-fields (i.e.\! operator-valued distributions on 
spacetime). Then to give an interactive model one should consider the 
restriction of the same free fields at a certain spacelike hyperplane but
then extend it to spacetime with a different {\it time-evolution}. (So 
that the interactive and free fields will coincide at our fixed spacelike 
hyperplane but possible nowhere else.)

However, Haag's theorem (see the book \cite{andallthat} for 
a detailed account) has ruled out the existence of such a description.
Suppose two QFT models are given: one in terms of the Wightmann-fields  
$\Phi_r\; (r=1,\ldots, n)$ and another one in terms of the 
Wightmann-fields $\tilde{\Phi}_r\; (r=1,\ldots, n)$. Assuming some 
relatively mild conditions (such as the existence of well-behavied
restrictions for the fields along spacelike hyperplanes), if there 
exists a spacelike hyperplane $H$ and a unitary operator $V$ such that 
\begin{equation}
\label{eq:V}
V\Phi_r(x)V^* = \tilde{\Phi}_r(x)\;\;\;\;\;\; (x\in H, r=1,\ldots n),
\end{equation}
then it also follows that up to a possible phase-factor, 
$V\Omega=\tilde{\Omega}$ (where $\Omega$ and $\tilde{\Omega}$ are the 
respective vacuum-vectors), $V\Phi_r(x)V^* = \tilde{\Phi}_r(x)$ for all 
spacetime points $x$ and $r=1,\ldots, n$, and finally, that 
$VU(g)V^*=\tilde{U}(g)$ (where $U$ and $\tilde{U}$ are the respective 
representations of spacetime symmetries) for all elements $g$ of the 
connected Poincar\'e group. Thus $V$ establishes an equivalence 
between the two models: if one was free, so is the other --- 
we cannot make an interacting model out of a free one in this way. 

So what would be an algebraic version of Haag's theorem?
Fix a spacelike hyperplane $H$. We shall say that two nets 
of von Neumann algebras $\A$ and $\tilde{\A}$ are {\bf equivalent 
along} $H$, iff there is a unitary operator $V$ such that
\begin{equation}
\label{eq:V_alg}
V\A(K^\diamond)V^* = \tilde{\A}(K^\diamond)
\end{equation}
for every double-cone $K^\diamond$ with base $K\subset H$. (See section 
\ref{sec:prel} on the notions of double-cones and algebraic QFT.)
In this paper --- under certain additional assumptions of ---
it will be proved that if $(\A,U)$ and $(\tilde{\A},\tilde{U})$ are two 
algebraic QFT models and $\A$ and $\tilde{A}$ are equivalent along a 
spacelike hyperplane, then there exists a unitary $W$ such that 
$W\A({\mathcal O})W^* = \tilde{\A}({\mathcal O})$ for all double-cones ${\mathcal O}$ and 
$WU(g)W^* = \tilde{U}(g)$ for all elements $g$ of the connected 
Poincar\'e group: thus the two models are equivalent. 

Actually, it will suffice to assume equivalence along a 
``half-hyperplane'' $H^+$; see 
the details in section \ref{sec:main}. Nevertheless, $H^+$ has still
an infinite space-volume. In fact, it is well-known that two inequivalent 
models, when restricted to a compact region, may give rise to 
unitarily equivalent nets of algebras; see \cite{glimmjaffe,EF} for 
examples.

\subsection{Algebraic vs.\! original version}

Of course in a strict sense the two versions 
of Haag's theorem cannot be compared. They are statements made in two 
different frameworks and despite numerous 
attempts, the passage between the two frameworks --- albeit clear in 
actual examples --- in general is still unresolved. 

Nonetheless, in some sense, as we shall see now, one may say that the 
algebraic version is stronger than the original one, and that the new 
version is not a simple reformulation of the old one. Let us see why. 

In case we deal with algebraic QFT 
models associated to Wightmann field theories, our additional assumptions 
--- with the exception of split property, which however probably could be 
avoided (see the comments at the end of section \ref{sec:prel}) --- are 
known to hold. To appreciate the differences, rather than at 
assumptions regarding frameworks, one must look at the 
respective notions of equivalence and the ways in which it is 
established.

The natural notion of equivalence of Wightmann field theories (i.e. the 
existence of a unitary operator connecting the defining fields and 
representations) --- and hence also the condition of equivalence along a 
spacelike hyperplane appearing in Haag's original version --- is too 
narrow, and does not coincide with physical equivalence. (In a sense 
this was exactly the original motivation \cite{haagbook} for considering 
the local algebras generated by the fields rather than the fields 
themselves: they already contain all physical information --- 
fields also depend on the choices made regarding our {\it description}.)

But there is more to this.
In the original version, the unitary operator $V$ appearing in 
equation (\ref{eq:V}) actually also turns out to be the unitary operator 
establishing the equivalence between the two models. This clearly does 
{\it not} hold in the algebraic case. 

For example, let both models be the same scalar free field model. Since 
the adjoint action of a Weyl-operator $W(f)$ preserves every local 
algebra, $V:=W(f)$ satisfies the requirement (\ref{eq:V_alg}) made in the 
algebraic version. However, in general $W(f)\Omega \neq \lambda \Omega$ 
so $V=W(f)$ does not establish an equivalence between the 
model and itself. To put it another way: a unitary operator whose adjoint 
action leaves the fields along a hyperplane invariant must be a multiple 
of the identity and hence must preserve every local algebra. To the 
contrary, a unitary operator, whose adjoint action preserves every local 
algebra, does not necessarily preserve the vacuum and hence may not take a 
Poincar\'e-covariant field into a Poincar\'e-covariant field.

So even if we the passage between Wightmann field theory and 
algebraic QFT was clear, the introduced algebraic version of Haag's 
theorem would not become a simple consequence of the original one. 
Rather, it is the other way around.

\subsection{Conformal QFT and half-sided modular 
inclusions}

Though it is always nice to strengthen a theorem, this was 
not why the author considered an algebraic version of Haag's theorem. 
As it will be explained now, original motivation came from conformal 
chiral QFT and in particular its relation to half-sided modular 
inclusions.

M\"obius covariant nets on $S^1$ have remarkable properties. Many things 
that in ``ordinary'' algebraic QFT often appear as additional assumptions 
--- like for example {\it additivity, Bisognano-Wichmann property} and 
{\it factoriality} of local algebras --- can be in fact {\it derived}; see 
\cite{FrG,BGL,GL96,FrJ,GLW} on the general structure of such nets.

For simplicity of notations, let us consider such a net $\A$ with vacuum 
vector $\Omega$ on the real line $\RR$ (see the last section on details of 
what it exactly means). Setting $\M:=\A(0,\infty)$ and 
$\N:=\A(1,\infty)$,
by an application of the {\it Bisognano-Wichmann property} (which, as was 
mentioned, in the conformal case is automatic) 
we have that the $(\Omega, \N\subset \M)$ is a {\bf standard half-sided 
modular inclusion} of von Neumann factors. That is, 
\begin{itemize}
\item
$\Omega$ is a {\it standard vector} of the inclusion $\N\subset \M$: it 
is cyclic and separating for both $\N,\M$ and $\N'\cap \M$, 
\item
$\Delta_{\M,\Omega}^{it}\N\Delta_{\M,\Omega}^{it}\subset \N$
for all $t\leq 0$.
\end{itemize}
This also works the other way around. 
Namely, it is shown \cite{Wie93,AZs,Wie94,GLW} that if $(\Omega, \N\subset 
\M)$ is a standard half-sided modular inclusion of factors, then one can 
construct a unique {\it strongly additive} M\"obius covariant net $\A$ 
with vacuum vector $\Omega$ such that $\A(0,\infty)=\M$ and 
$\A(1,\infty)=\N$. 

At first sight, this seems to give a great opportunity for constructing 
new conformal chiral QFT models. Indeed, instead of an entire net of 
algebras (together with a representation of the M\"obius group), all we 
need is to present a certain standard inclusion of von Neumann factors.

Sadly, the reality is just the opposite way around. As far as the author 
knows, (nontrivial) standard half-sided modular inclusions have been 
constructed only with the help of M\"obius covariant nets.
However, there were hopes to find a more or less direct way to 
construct a new half-sided modular inclusions out of an existing one. 
R.\! Longo proposed\footnote{This idea has never been published; 
the author learned about it through oral communication.} to consider the 
following ``perturbation'' of a half-sided modular inclusion of factors 
$(\Omega,\N\subset\M)$. 

For a vector $\Psi$ which is cyclic and separating for $\M$, let us denote by 
$J_\Psi$ and $\Delta_\Psi$ the modular objects associated to $(\Psi,\M)$. 
By \cite{araki}, for each $X\in\M$, $X^*=X$ there exists a  
vector $\Omega_X$ cyclic and separating for $\M$ such that
\begin{itemize}
\item $\Omega_X$ is in the natural cone of $(\Omega,\M)$
and hence $J_{\Omega_X}=J_\Omega=:J$,
\item ${\rm ln}(\Delta_{\Omega_X})= {\rm ln}(\Delta_\Omega)+ X + JXJ$.
\end{itemize}
In particular, if $X\in\N$ then by applying the {\it 
Trotter product formula} one can easily check that 
$(\Omega_X,\N\subset \M)$ is still a half-sided modular inclusion. If 
$\Omega_X$ is also a standard vector for $\N\subset \M$, then we can go on 
and generate a new strongly additive net $\A_X$. 

But are the original net $\A_0$ (from where we took our half-sided 
modular inclusion) and $\A_X$ really different?
Using the mentioned product formula
one can also easily show that with $X\in\N$ many local algebras will 
remain the same; not only that $\A_0(0,\infty)=\A_X(0,\infty)= \M$ and 
$\A_0(1,\infty)=\A_X(1,\infty)=\N$ but actually 
\begin{equation}
\label{eq:localeq}
\A_0(I)= \A_X(I) \;\;\;\;\textrm{for all }\, I\subset (0,1). 
\end{equation}
On the other hand, by an easy reformulation (see section 
\ref{sec:conformal}) of the main result of the present paper, if $\A_0$ 
and hence also $\A_X$ satisfy the {\it split property}, then the above 
equality implies that $\A_0$ and $\A_X$, as 
M\"obius covariant nets, are equivalent. Thus, in this way we cannot 
obtain new models.

Of course one may try to improve the situation. Instead of a self-adjoint 
$X\in\N$, more generally we could take any $X\in\M$, $X^*=X$ for which 
$e^{iXt}\N e^{-iXt}\subset \N$ for all $t\leq 0$. For example, $X$ may be 
a self-adjoint of the form $X=X_1 + X_2$ with $X_1\in \N$ and $X_2\in\M\cap 
\N'$, and in concrete examples we may find further choices. 

Neveretheless, 
in light of Haag's theorem, it seems unlikely for the author that 
retaining the same inclusion $\N\subset \M$ and changing only the 
``dynamics'' one could obtain something really new. 
Actually in section \ref{sec:conformal}, regarding this question
we shall observe two important facts.
Let $(\Omega,\N\subset \M)$ and $(\tilde{\Omega},\tilde{\N}\subset 
\tilde{\M})$ be two standard half-sided modular inclusions of factors and 
denote the two corresponding strongly additive M\"obius covariant nets 
by $(\A,U)$ and $(\tilde{\A},\tilde{U})$, respectively.
\begin{itemize}
\item[I.] If there exists a unitary operator $V$ such that 
$V\N V^*= \tilde{\N}$ and $V\M V^*=\tilde{\M}$ then for each 
$n\in\NN$ there exists a unitary operator $V_n$ such that
$$
V_n\A(j,k)V_n^* = \tilde{\A}(j,k)
$$
for every pair of integers $j,k\in\{0,1,\ldots ,n\}$.
\end{itemize}
In particular, this implies that if $\A$ is split, so is $\tilde{\A}$
and in fact $\A$ and $\tilde{\A}$ has unitarily equivalent 
2-{\it interval inclusions}. Now this inclusion is a rich source of 
information; in the {\it completely rational} case essentially it contains 
\cite{KLM} the entire representation theory of the net. So this already 
suggests that
perhaps $\A$ and $\tilde{\A}$ are equivalent. As a matter of fact, 
the conformal version of our algebraic Haag's 
theorem tells that just a slightly stronger condition indeed
implies equivalence.
\begin{itemize} 
\item[I\!I.] Let $(\A,U)$ and $(\tilde{\A},\tilde{U})$ be two M\"obius 
covariant nets and assume that at least one of them is split. If there 
exists a unitary operator $V$ such that
$$
V\A(j,k)V^* = \tilde{\A}(j,k)
$$
for every pair of natural numbers $j,k \in\NN$ then $(\A,U)$ and 
$(\tilde{A},\tilde{U})$ are equivalent.
\end{itemize}
Again, as was mentioned already and will be explained at the end of the 
next section, the author thinks that split condition should be possible to 
remove. Now I + I\!I + the remarks made after stating them  --- though 
does not actually prove --- seems to indicate the following.
\begin{conjecture}
The unitary equivalence class 
of a standard half-sided modular inclusion of 
factors $(\Omega,\N\subset\M)$ is completely determined
(up to a possible normalization of $\Omega$)
by the unitary equivalence class of the inclusion $\N\subset \M$. 
That is, for another
half-sided modular inclusion 
$(\tilde{\Omega},\tilde{\N}\subset \tilde{\M})$ 
with equal normalization $\|\tilde{\Omega}\| = \|\Omega\|$, if 
there 
exists a unitary operator $V$ such that $V\N V^*=\tilde{\N}$ and 
$V\M V^* = \tilde{\M}$, then there exists a unitary operator $W$ such that
not only $W\N W^* = \tilde{\N}$ and $W\M W^* = \tilde{\M}$, but also 
$W\Omega = \tilde{\Omega}$.
\end{conjecture}

\section{Preliminaries: axioms of algebraic QFT}
\label{sec:prel}

In this paper we shall consider an algebraic version of Haag's theorem. An 
algebraic QFT, rather than quantum fields, is given in terms of a net of 
local algebras ${\mathcal O}\mapsto \A({\mathcal O})$. We shall work directly on the so-called 
``vacuum Hilbert space'' and consider ${\mathcal O}\mapsto \A({\mathcal O})''=\A({\mathcal O})$ to be a 
net of {\it von Neumann algebras}. 

For a spacelike hyperplane $H$, and a bounded, connected and simply 
connected open subset $K$ of $H$ we set
\begin{equation}
K^\diamond:= {\rm Int}(\overline{K}^c),
\end{equation}
where $K^c$ is the (closed) {\it causal completion} of $K$ and ``${\rm 
Int}$'' stands for the (open) interiour. We say that $K^\diamond$ is a 
{\bf double-cone} with {\bf base} on $H$; that is, with $K\subset H$.

For physical purposes (e.g.\! for determining the $S$-matrix or the 
structure of charged sectors) it is enough to work with special 
spacetime regions like double-cones. So considering only what 
is absolutely necessary, here we define an {\bf algebraic QFT} to 
be a map associating to each double-cone ${\mathcal O}$ a von Neumann algebra $\A({\mathcal O})$ 
on a fixed Hilbert space $\H$ together with a strongly continuous unitary 
representation $U$ of the connected Poincar\'e group satisfying the 
following ``minimal'' conditions. (Note that some further additional 
properties will be later 
considered.)
\begin{itemize}
\item[(1)] {\bf Isotony}: $\A({\mathcal O}_1)\subset \A({\mathcal O}_2)$ 
whenever
${\mathcal O}_1\subset{\mathcal O}_2$.
\item[(2)] {\bf Locality}: $[\A({\mathcal O}_1),\A({\mathcal O}_2)]=0$ whenever ${\mathcal O}_1$
and ${\mathcal O}_2$ are spacelike separated.
\item[(3)] {\bf Covariance}: $U(g)\A({\mathcal O})U(g)^*=\A(g({\mathcal O}))$
for all regions ${\mathcal O}$ and elements $g$ of the connected Poincar\'e group.
\item[(4)] {\bf Positivity of energy}: $P_{\mathbf x}\geq 0$ whenever 
$\mathbf x$ is future like. $P_{\mathbf x}$ is defined by the equation 
$U(\tau_{t \mathbf x}) = e^{itP_{\mathbf x}}\; (t\in\RR)$ in which 
$\tau_{\mathbf z}$ is a translation by $\mathbf z$.
\item[(5)] {\bf Existence, uniqueness and cyclicity of vacuum}: 
up to phase there exists a unique unit vector $\Omega$ invariant for
$U(\tau)$ for all spacetime translations $\tau$. Moreover, 
$\Omega$ is cyclic for $\vee_{{\mathcal O}}\A({\mathcal O})$.
\end{itemize}
Note that from a physical point of view one should assume $U$ to be  
a {\it projective} representation rather than a true one. However,
it is easy to see that if $U$ is a projective representation of a group $G$ 
and $N\subset G$ is a normal subgroup such that there exists a unique 
one-dimensional invariant subspace for $U(N)$, then actually 
this subspace is invariant for the action of the full group and 
hence one can arrange the ``phase factors'' in 
such a way that $U$ becomes a true representation. So without loss of 
generality, for clarity we have stated the axioms with $U$ being a true 
representation rather than just a projective one.

Although so far we have only associated algebras to double-cone like 
bounded regions, by setting
\begin{equation}
\A({\mathcal O}):= \vee_{K^\diamond\subset {\mathcal O}}\A({\mathcal O})
\end{equation}
we may talk about the algebra $\A({\mathcal O})$ associated to any open region ${\mathcal O}$.
Note that {\it isotony} implies that the new definition does not change 
the algebra associated to a double-cone and that properties (1,2,3,4,5) 
remain valid.

The standard {\it Reeh-Schlieder} argument combined with {\it 
locality} shows that $\Omega$ is cyclic and separating for $\A(\W)$ 
whenever $\W$ is a {\it wedge region}. (See e.g.\! the book 
\cite{haagbook} for precise definition of wedge regions.) Actually, 
by \cite[Thm.\! 3]{longo} it even follows that for a wedge region 
$\W$ the algebra $\A(\W)$ is a type I\!I\!I$_1$ factor.
Another well-known consequence of (1,2,3,4,5) is 
{\it irreducibility}:
\begin{equation}
\A(M)'\equiv \{\vee_{{\mathcal O}}\A({\mathcal O})\}' = \cap_{{\mathcal O}}\A({\mathcal O})' = \CC\mathbbm 1.
\end{equation}
Here $M$ stands for the full spacetime.

Howeve, as was mentioned, (1,2,3,4,5) is only a ``minimal set'' of 
conditions; they still allow many pathological examples. In particular, 
while $\Omega$ turns out to be cyclic for $\A(\W)$ whenever $\W$ is a 
wedge, it may not be so for a double-cone\footnote{
For an example, let us fix a bounded open set of spacetime and call a 
region to be ``small'' if 
it can be moved into this set by a Poincar\'e tranformation. Now take a 
``nice'' model and reset all local algebras that are associated to ``small 
region'' to be equal to the trivial algebra $\CC\mathbbm 1$. It is easy 
to see that all listed properties remain valid, but now, starting from a 
``nice'' model we have produced one with the mentioned pathological 
property.}.

Sometimes instead of {\it isotony} the stronger {\bf additivity} property 
is required; namely, that $\A({\mathcal O})\subset \vee_{k=1}^n \A({\mathcal O}_k)$ whenever 
${\mathcal O}\subset \cup_{k=1}^n{\mathcal O}_k$. (Note that in the conformal case 
additivity is not needed as a further assumption since it can be  
actually {\it proved}, as it will be discussed in section 
\ref{sec:conformal}.)
Having additivity one can use the argument of 
Reeh and Schlieder and show that $\Omega$ is cyclic for every local 
algebra $\A({\mathcal O})$ associated to a nonempty open region ${\mathcal O}$.

Local von Neumann algebras were originally introduced to replace the 
unbounded polinomial algebra of local fields. From a physicist point of 
view it seems reasonable to assume that our local von Neumann algebras are 
in fact generated by unbounded (Wightmann) fields (i.e.\! that there is an 
``underlying'' Wightmann field theory and $\A({\mathcal O})$ is the smallest von 
Neumann algebra to which the closure of all fields smeared with 
testfunctions with support in ${\mathcal O}$ are affiliated). Now for the algebra of 
fields additivity is evident. However, the passage from unbounded 
operators to von Neumann algebras 
is a delicate issue. In particular --- at least, up to the knowledge of 
the author --- even assuming an underlying Wightmann field 
theory, so far additivity could not be proved. On the other hand, it is 
easy to see that the cyclicity guaranteed (at the level of Wightmann 
fields) by the Reeh-Schlieder theorem passes without problems to the 
level of local von Neumann algebras. For this reason here we shall 
assume directly this cyclicity rather than making the stronger assumption 
of additivity.
\begin{itemize}
\item[(6)] {\bf Reeh-Schlieder property}: 
$\Omega$ is cyclic for $\A({\mathcal O})$ for every nonempty open region ${\mathcal O}$.
\end{itemize}

Let $\W$ be a wedge region and consider the modular operator 
$\Delta_{\A(\W),\Omega}$ and modular conjugation $J_{\A(\W),\Omega}$
associated to $(\A(\W),\Omega)$. (As was mentioned, $\Omega$ is cyclic and 
separating for $\A(\W)$, so these objects are well-defined). Assuming the 
existence of an ``underlying'' Wightmann field theory, it is known 
\cite{biwi} that these objects have a ``geometrical meaning''. Though 
attempts were made, so far it has not been 
proved that in general, a geometrical nature of these modular objects is a 
consequence of (1,2,3,4,5). So we shall simply assume it.
\begin{itemize}
\item[(7)] {\bf Bisognano-Wichmann property}. If $\W$ is a 
{\it wedge-region}, then $\Delta_{\A(\W),\Omega}^{it}= U(\beta_t)$ 
$(t\in\RR)$
where $t\mapsto \beta$ is the one-parameter group of {\it boosts}
associated to $\W$ with a certain parametrization.
\end{itemize}
For definition of the one-parameter group of boosts associated
to a wedge and details on the parametrization we refer to the book 
\cite{haagbook}. Note that as it will be explained in section 
\ref{sec:conformal}, in the conformal case not only (6), but 
also this property can be derived eliminating the need to 
additionaly assume it.  

The discussed properties (1,2,3,4,5,6,7) are essential for 
the proof of our argument. However --- though for a somewhat technical 
reason --- we shall actually need one more thing.
\begin{itemize}
\item[(8)] {\bf Split property}. The
 inclusion $\A(K^\diamond_1)\subset \A(K^\diamond_2)$ is split whenever 
$\overline{K_1^\diamond}\subset K^\diamond_2$.
\end{itemize}
(Actually {\it distant split property} would suffice for us, but for 
simplicity here we only talk about split property.) For physical 
significance of the split property we again refer to the book 
\cite{haagbook}. Here we briefly comment only on the {\it difference} 
between how (8) and the other properties will be used. 

In the course of the proof of our main theorem, we shall construct a 
sequence of unitary operators $n\mapsto W_n$. The equivalence between the 
two models is then to be established by the strong limit of this sequence. 
But though the author is convienced that this limit exists, he could 
not show this. So instead, split property is used to obtain a compactness 
condition by which at least the existence of a convergent subsequence can 
be established.

Now the way in which split property can be turned into the right 
compactness condition is not simple; in fact the whole next section will 
be dedicated to this question. Nevertheless, the author feels that split 
property should not play an essential role in the algebraic Haag's 
theorem.

\section{On split inclusions}

A $\N\subset \M$ be an inclusion of von Neumann algebras for which 
there exists a type I factor $\R$ ``in between'': $\N\subset \R\subset 
\M$, is said to be a {\bf split} inclusion. Let $\N\subset \M$ be a split 
inclusion and $\Omega$ a {\bf standard vector} for the inclusion in 
question; i.e.\! we suppose
that $\Omega$ is cyclic and separating for both $\N,\M$ and the relative 
commutant $\N'\cap \M$. Denoting the modular conjugation 
associated to $\N'\cap \M$ and the vector $\Omega$ by $J_\Omega$, 
if $\N$ is a factor, we shall set
\begin{equation} 
\R_\Omega: = \N \vee J_\Omega \N J_\Omega. 
\end{equation}
Alternatively, if $\M$ is a factor, we shall set
\begin{equation} 
\R_\Omega: = \M \cap J_\Omega \M J_\Omega.
\end{equation}
By \cite{DL}, under the assumptions made our notation is unambiguous: if 
both $\N$ and $\M$ are factors, then 
$\N \vee J_\Omega \N J_\Omega =  \M \cap J_\Omega \M J_\Omega$. 
Moreover, the thus defined von Neumann algebra $\R_\Omega$ is a type I 
factor between $\N$ and $\M$; we shall say that 
that $\R_\Omega$ is the {\bf canonical type I factor} of the inclusion.

If $(\Omega, \N\subset \M)$ is a standard split inclusion in which one of 
the algebras is a factor, and $W$ is a unitary operator such that it 
preserves the vector $\Omega$ and its adjoint action preserves the 
algebras $\N$ and $\M$, then the adjoint action of $W$ must also preserve 
the canonical type I factor $\R_\Omega$ of the inclusion. Using this fact 
in \cite{DL} it was proved that the group of such unitary operators is 
compact and metrizable (with respect to the strong operator topology). In 
particular, if $n\mapsto W_n$ is a sequence of such unitary operators, 
then one can always find a subsequence $s$ such that $n\mapsto W_{s(n)}$ 
will strongly converge to a unitary operator.

In this section we assume that $(\Omega,\N\subset\M)$ is a standard 
split inclusion in which at least one of the algebras is a factor, and 
$n\mapsto W_n$ is a sequence of unitaries such 
that the adjoint action of $W_n$ preserves the algebras $\N$ and $\M$ for 
all $n\in \NN$. We set $\Omega_n:=W_n\Omega$ and assume that $n\mapsto 
\Omega_n$ is convergent; more precisely, that there is a standard vector 
$\Psi$ for our inclusion such that
$\|\Omega_n - \Psi\|\to 0$ as $n\to \infty$.
Our aim is to find a suitable modification of
the proof of \cite{DL} in order to show the existence of a subsequence 
of $n\mapsto W_n$ converging strongly to a unitary operator.  

We shall proceed in several intermediate steps. We shall begin with an 
important observation which generalizes \cite[Lemma 3.2]{DL}.
\begin{lemma}
Let $\H$ be a Hilbert space, $n\mapsto U_n$ a sequence of 
unitary operators on $\H$ and $\varphi$ a faithful normal state on $\B(\H)$. 
If $n\mapsto \varphi_n:= \varphi\circ{\rm Ad}(U_n)$ converges 
in norm, then there exists a subsequence $s$ such that 
$n\mapsto U_{s(n)}$ converges strongly. Moreover, if the norm limit 
of $n\mapsto \varphi_n$ is faithful, then the strong limit of $n\mapsto 
U_{s(n)}$ is unitary.
\end{lemma} 
\begin{proof} 
If ${\rm dim}(\H)<\infty$, then the statement is trivially true. 
On the other hand, as $\B(\H)$ was assumed to 
have a faithful normal state, $\H$ must be separable. So we may assume 
that $\H$ is the (up to unitary equivalence) unique infinite 
dimensional separable Hilbert space.

For each normal state $\eta$ there exists a unique positive 
trace-class operator $D_\eta\in\B(\H)$ such that
\begin{equation}
\eta(A)= {\rm Tr}(D_\eta A)
\end{equation}
for all $A\in\B(\H)$. Since $\varphi_n= \varphi\circ{\rm Ad}(U_n)$, we 
have that $D_n:=D_{\varphi_n}=U_n D_\varphi U_n^*$. Now let 
$\tilde{\varphi}$ be the assumed (norm) limit of $n\mapsto \varphi_n$, 
and consider the operator $D_{\tilde{\varphi}}$. We know that 
$\varphi_n\to \tilde{\varphi}$ in norm as $n\to \infty$. What does this 
tell us about $D_n\; (n\in \NN)$ and $D_{\tilde{\varphi}}$? Since 
$\|D_n-D_{\tilde{\varphi}}\|\leq 2$, as $n\to \infty$ we have
\begin{equation}
0\leq {\rm Tr}((D_{\tilde{\varphi}}-D_n)^2) 
=({\tilde{\varphi}}-\varphi_n)
(D_\Psi-D_n)\leq 2 
\|{\tilde{\varphi}}-\varphi_n\|\to 0.
\end{equation}
In particular, $D_n\to D_{\tilde{\varphi}}$ in norm. 
Let us see now what can we say 
about the convergence of spectrums and spectral projections.

Let $f$ be a continuous real function on $[0,1]$
and $\epsilon>0$. Then by the Stone-Weierstrass theorem there is a real 
polinomial $p$ such that $|f(x)-p(x)|<\epsilon/3$ for all $x\in [0,1]$. 
As ${\rm Sp}(D_n),\, {\rm Sp}(D_{\tilde{\varphi}})\subset [0,1]$, we have that 
both $\|f(D_{\tilde{\varphi}})-p(D_{\tilde{\varphi}})\|<\epsilon/3$ and 
$\|f(D_n)-p(D_n)\|<\epsilon/3$ for all $n\in\NN$.
On the other hand, as $p$ is a polynomial and $D_n\to D_{\tilde{\varphi}}$ 
in norm, there exists a $N\in\NN$ such that 
$\|p(D_n)-p(D_{\tilde{\varphi}})\|<\epsilon/3$ for all $n>N$. Thus for 
$n>N$ we have that
\begin{equation}
\|f(D_n)-f(D_{\tilde{\varphi}})\| 
\leq
\|f(D_n)-p(D_n)\| + \|p(D_n)-p(D_{\tilde{\varphi}})\| + 
\|p(D_{\tilde{\varphi}})-f(D_{\tilde{\varphi}})\| 
< \epsilon
\end{equation}
showing that $\|f(D_n)-f(D_{\tilde{\varphi}})\|\to 0$ as $n\to \infty$.

As was already noted, ${\rm Sp}(D_n)={\rm Sp}(D_\varphi)$ for 
every $n\in\NN$ because of unitary equivalence.
Now $D_\varphi, D_{\tilde{\varphi}}$ are {\it density operators}, 
so their spectrum is contained in $[0,1]$ and have at most one point of 
accumulation; namely, at zero. Moreover, each positive point of their 
spectrum must be an eigenvalue corresponding to a finite dimensional 
eigenspace.

Since the spectrum is compact, if $x\notin {\rm 
Sp}(D_\varphi)$, then there exists a 
continuous function $f$ on $[0,1]$ such that $f|_{{\rm Sp}(D_\varphi)}=0$ 
but $f(x)=1$. Thus
$f(D_\varphi)=0= f(D_n)$ so
by the established convergence property $f(D_{\tilde{\varphi}})=0$
showing that $x$ cannot be an eigenvalue for $D_{\tilde{\varphi}}$.
On the other hand, let us fix an eigenvalue $\lambda\in{\rm 
Sp}(D_\varphi)\setminus \{0\}$ and choose a continuous
function $f$ on $[0,1]$
such that $f(x)=0$ for all $x\in{\rm
Sp}(D_\varphi)\setminus \{\lambda\}$  
and $f(x)=1$ if and only if $x=\lambda$.
Then $f(D_n)$ is exactly the spectral projection associated 
to the eigenvalue $\lambda$ of $D_n$; in particular
$\|f(D_n)\|=1$ and $f(D_n)^2=f(D_n)^*=f(D_n)$. This shows that
$f(D_{\tilde{\varphi}})$ --- which is the norm-limit of $f(D_n)$
--- is also a nonzero projection.

Now $0\in {\rm Sp}(D_{\varphi})\cap {\rm Sp}
(D_{\tilde{\varphi}})$ since we are dealing with 
density operators given on an infinite-dimensional 
space. Moreover, $0$ is not an eigenvalue for $D_\varphi$
since $\varphi$ was assumed to be faithful.

Let us sum up what we have obtained so far. We have shown 
that ${\rm Sp}(D_\varphi)={\rm Sp}(D_n)={\rm Sp}(D_{\tilde{\varphi}})$
and that the for each eigenvalue $\lambda$ of $D_\varphi$, the 
spectral projections $E_{n,\lambda}$
of $D_n$ corresponding to the eigenvalue $\lambda$ converge
in norm to the spectral projection $E_{\tilde{\varphi},\lambda}$
of $D_{\tilde{\varphi}}$ corresponding to the same eigenvalue $\lambda$.

Let $\Phi$ be an eigenvector of $D_\varphi$ with eigenvalue 
$\lambda$. Then $\Phi_n:=U_n\Phi$ is an eigenvector of 
$D_n=U_n D_\varphi U_n^*$ with the same eigenvalue. To put it in another 
way, $(E_{n,\lambda}-\mathbbm 1) \Phi_n = 0$ implying that  
$\|(E_{\tilde{\varphi},\lambda} -\mathbbm 1)\Phi_n \|\to 0$ as $n\to 
\infty$ and so $n\mapsto\Phi_n$ (or a subsequence of it) converges if and 
only if $n\mapsto E_{\tilde{\varphi},\lambda}\Phi_n$ (or its subsequence 
in question) 
does so. 
But $n\mapsto E_{\tilde{\varphi},\lambda}\Phi_n$ ``runs'' in the unit ball 
of the 
finite dimensional space ${\rm Im}(E_{\tilde{\varphi},\lambda})$, so it 
admits a convergent subsequence.

Since $D_{\varphi}$ is a density operator, there exists a complete 
orthonormal system consisting of eigenvectors of $D_{\varphi}$, only. 
However, since $\H$ is separable, this system is countable. Thus by what 
was established, we may conclude the existence of a subsequence $s$ such 
that $n\mapsto U_{s(n)}$ strongly converges on each vector of this 
system, and hence --- as we are dealing with a sequence of unitary 
operators --- on every vector of $\H$. 

We are almost finished: we have proved the existence of a convergent 
subsequence. However, the limit of a strongly convergent 
sequence of unitary operators may not be again a unitary operator (in 
general, it is only an isometry). To show the existence of a unitary 
limit, we have to check the strong convergence of the adjoints. Then
to conclude our proof, all we have to note is that if $\tilde{\varphi}$ 
is faithful, then we may repeat our 
argument with the unitary sequence $n\mapsto U^*_{s(n)}$ and with the 
role of $\varphi$ and $\tilde{\varphi}$ exchanged.
\end{proof}

Recall that in this section we are dealing with a standard split inclusion
$(\Omega, \N\subset \M)$ in which at least one of the algebras is a 
factor, and a certain sequence of unitary operators $n\mapsto W_n$. In our 
case the adjoint action of $W_n$ does not 
necessarily preserve the canonical type I factor $\R_\Omega$. Rather,
we have that $W_n\R_\Omega W_n^*=\R_{\Omega_n}$ where $\R_{\Omega_n}$
is the canonical type I factor given by the vector $\Omega_n$. (Note that
$\Omega_n = W_n\Omega$ is automatically a standard vector for the 
inclusion $\N\subset \M$). All we can 
hope now that since $\Omega_n\to \Psi$, the type I factors $\R_{\Omega_n}$ 
will get ``closer and closer'' to the type I 
factor $\R_\Psi$. At this point, our previous lemma resolves
only the rather particular case when the adjoint action of $W_n$ 
actually {\it does} preserve $\R_\Omega$. However, this in turn will serve 
to prove the general case.
\begin{proposition}
Suppose that for all $n\in \NN$, the adjoint action of $W_n$ also 
preserves the canonical type I factor $\R_\Omega$. Then there exists a 
subsequence $s$ such that $n\mapsto W_{s(n)}$ strongly 
converges to a unitary operator.
\end{proposition}
\begin{proof}
First let us note that if $n\mapsto A_n$ is a sequence of uniformly 
bounded operators converging strongly to a bounded operator $A$ then
also $A_n\otimes \mathbbm 1 \to A\otimes \mathbbm 1$ strongly, as $n\to 
\infty$. Indeed, convergence is clear on vectors of tensorial form, and 
hence on every vector as our sequence was assumed to be uniformly bounded.
In particular, if we identify $\R_\Omega$ with $\B(\K)$ (via an 
isomorphism) where $\K$ is some Hilbert space, then a sequence of unitary 
operators $n\mapsto U_n\in\R_\Omega$ is strongly converging to unitary 
operator of $\R_\Omega$ if and only if we have convergence in the topology 
given by the strong operator topology of $\B(\K)$.

So let now $\omega$ and $\psi$ be the normal states on $\R_\Omega$ given 
by the vectors $\Omega$ and $\Psi$, respectively. These states are 
faithful since the vectors in question are separating for $\M$ which 
contains $\R_\Omega$. 

Since $\R_\Omega$ is a type I factor, the adjoint action of $W_n$ 
in $\R_\Omega$ can be implemented by a unitary $U_n\in \R_\Omega$. 
We have that $\omega \circ{\rm Ad}(U_n)\to \psi$ in norm, since
$\|W_n\Omega-\Psi\|\to 0$. Thus our previous lemma can be applied, 
and by what was noted in the beginning of our proof, it shows that 
there exists a unitary $U\in\R_\Omega$ and a subsequence $s$ such that 
$U_{s(n)}\to U$ strongly (on our original Hilbert space, not only on $\K$)
as $n\to \infty$. Then for an $A\in \R_\Omega$ we have that 
as $n\to \infty$,
\begin{equation}
W_{s(n)} A \Omega = (U_{s(n)} A U_{s(n)} ^*) W_{s(n)} \Omega \to U A U ^* 
\Psi
\end{equation}
since $\|W_{s(n)}\Omega - \Psi\|\to 0$ and since the strong limit of a 
product of strongly convergent, uniformly bounded sequences is simply the 
product of the limits. Thus $n \mapsto W_{s(n)}$ is strongly convergent on 
$\overline{\R_\Omega\Omega}$, and $n \mapsto W_{s(n)}^*$ is strongly 
convergent on $\overline{\R_\Omega\Psi}$. Now both $\Omega$ and $\Psi$
are cyclic for $\N$ and hence for $\R_\Omega$, too; so actually we 
have shown that $n \mapsto W_{s(n)}$ converges strongly to a unitary 
operator.
\end{proof}

In our previous proposition we assumed $\R_n:=\R_{\Omega_n}$ to coincide 
with $\R_\Omega$. It is rather clear that this assumption is too strong; 
it will not hold in general. So now we shall see how we can ``correct'' 
$W_n$ by another unitary in order to have this property.

\begin{lemma}
Let $\Psi,\Psi_n\; (n\in\NN)$ be standard vectors for the 
split inclusion $\N\subset \M$ in which at least one of the algebras is a 
factor. If $\|\Psi_n-\Psi\|\to 0$
as $n\to \infty$, then there exists a sequence of unitaries $n\mapsto 
U_n\in \N'\cap \M$ strongly converging to the operator $\mathbbm 1$ such
that $U_n\R_{\Psi_n}U_n^* = \R_\Psi$ for all $n\in\NN$.
\end{lemma}
\begin{proof}
We may assume that the smaller algebra $\N$ is a 
factor. (If only $\M$ is a factor, then instead of the original inclusion
we may consider $(\Omega, \M'\subset \N')$ in which it is again the 
smaller algebra which is a factor.)
Let us denote by $\psi,\psi_n$ $(n\in \NN)$ the faithful normal states
on $\N'\cap \M$ given by the vectors
$\Psi,\Psi_n$ $(n\in \NN)$, respectively. The state $\psi_n$ has a
unique vector-representation
$\tilde{\Psi}_n$ in the natural cone of $(\Psi, \N'\cap \M)$. Note
that by construction, the modular conjugation
$J_{\tilde{\Psi}_n}$ associated to $(\tilde{\Psi}_n,\N'\cap \M)$
coincides with $J_{\Psi}$.
As both cyclic and separating vectors
$\Psi_n$ and $\tilde{\Psi}_n$ implement the same state on $\N'\cap
\M$, there exists a unitary
$U'_n\in (\N'\cap \M)'$ such that $U'_n\Psi_n = \tilde{\Psi}_n$, or
equivalently, that
$U'^*_n\tilde{\Psi}_n = \Psi_n$. As the adjoint action of $U'_n$
preserves both $\N,\M$ and $\N'\cap \M$, we have that
\begin{equation}
U'^*_n J_{\Psi} U'_n = U'^*_n J_{\tilde{\Psi}_n} U'_n  =
J_{U'^*_n\tilde{\Psi}_n} = J_{\Psi_n}.
\end{equation}
Moreover, as rather evidently $J_{\Psi}U'^*_n \N U'_n J_{\Psi} \subset
\N'\cap \M$, we also
have that $U'_n$ is in the commutant of $J_{\Psi}U'^*_n \N U'_n
J_{\Psi}$
and
\begin{eqnarray}
\nonumber
J_{\Psi_n}\N J_{\Psi_n} &=& U'^*_n  J_{\Psi} U'_n \N U'^*_n
J_{\Psi} U'_n =
J_{\Psi} U'_n \N U'^*_n J_{\Psi}
\\
&=&
(J_{\Psi} U'_n J_{\Psi}) J_{\Psi} \N J_{\Psi} (J_{\Psi}U'_n
J_{\Psi})^*
=
U_n J_{\Psi_n} \N J_{\Psi_n} U^*_n,
\end{eqnarray}
where $U_n=J_{\Psi} U'_n J_{\Psi}$ is a unitary in the relative 
commutant $\N'\cap \M$.
Now the sequence of states $n\mapsto \psi_n$ clearly 
converges to $\psi$ in norm
(since $\Psi_n\to \Psi$ as $n\to \infty$). It follows that the
distance between the vectors $\tilde{\Psi}_n$ and $\Psi$, both elements 
of the the natural cone of $(\Psi, 
\N'\cap \M)$, also goes to zero as $n\to \infty$.
Now
$U_n\Psi= J_{\Psi}U'_n J_\Psi\Psi = J_{\Psi}U'_n \Psi$
and $\|J_{\Psi}U'_n \Psi - J_{\Psi}U'_n \Psi_n\|  = 
\|\Psi-\Psi_n\|\to 0$,
so $n\mapsto U_n\Psi$ is convergent as in fact
\begin{equation}
\lim_n(U_n\Psi) = \lim_n(J_\Psi U'_n\Psi_n) =
\lim_n (J_\Psi\tilde{\Psi}_n)= J_\Psi \Psi = \Psi 
\end{equation}
Since $U_n\in \N'\cap \M \subset \M$, the above shows that
$n\mapsto U_n$ strongly converges to the identity on the closure 
of $\M'\Psi$ and hence everywhere (as $\Psi$ is cyclic and separating for 
$\M$ and so for $\M'$, too).
\end{proof}
\begin{corollary}
\label{main_cor}
Under the assumptions explained in the beginning of this section, 
it follows that there exists a subsequence $s$ such that $n\mapsto 
W_{s(n)}$ strongly converges to a unitary operator.
\end{corollary}

\section{Equivalence of models}
\label{sec:main}

Fix a space-like hyperplane $H$, and 
let further $\tau$ be a nonzero translation such that $\tau(H)=H$.
Fix a plane $N$ in $H$ which is orthogonal to the direction of the 
translation $\tau$. 
Then $H\setminus N$ is the disjoint union of two open ``half-spaces'' 
$H^+$ and $H^-$. Here the ``$+$'' and ``$-$'' signs are given in such a 
way that $\tau(H^+)\subset H^+$ while $\tau(H^-)\supset H^-$. 

Note that 
$\W^\pm:=(H^\pm)^\diamond$ are wedge-regions
such that the causal complement of any of them is exactly (the closure of) 
the other. Moreover, we have
$\tau(\W^+)\subset \W^+$ and 
$\tau(\W^-)\supset \W^-$, 
and that $\cup_{n\in\NN} \tau^n(\W^-)$ is the full spacetime. Hence 
if $(\A,U_\A)$ is an algebraic QFT given on the Hilbert space $\H$
satisfying axioms (1,2,3,4,5) discussed 
in the introduction, then $n\mapsto \A(\tau^n(\W^-))$ is an increasing sequence 
of von Neumann algebras such that its union is dense  
(w.r.t.\! the strong op.\! topology) in $\B(\H)$. Then 
for the decreasing sequence $n\mapsto \A(\tau^n(\W^+))$, 
by {\it locality} we have that $\cap_{n\in\NN} \A(\tau^n(\W^+))=
\CC\mathbbm 1$.

In our main theorem --- apart from many other things --- we shall also use 
a rather well-known fact concerning a decreasing sequence of von Neumann 
algebras and distances of restrictions of states. However, for reasons of 
self-containment we shall outline the proof of this fact (which is anyway 
short).
\begin{lemma}
Let $\M_1\supset \M_2\supset\M_3\supset \ldots$ be a decreasing sequence 
of von Neumann algebras on a Hilbert space $\H$ with $\cap_{n\in\NN} 
\M_n=\CC\mathbbm 1$ (or equivalently: with 
$\{\cup_{n\in\NN}\M_n'\}''=\B(\H)$). Let further $\psi, \tilde{\psi}$ be 
two normal states on $\M_1$. Then for the restriction of states $\psi_n:= 
\psi|_{\M_n},\, \tilde{\psi}_n:=\tilde{\psi}|_{\M_n}$ we have
$$
\|\psi_n-\tilde{\psi}_n\| \to 0
$$
as $n\to \infty$
\end{lemma}
\begin{proof}
Clearly, the validity of the statement does not depend on the  
``underlying'' Hilbert space. So we may assume that both states 
on $\M_1$ can be represented by vectors in $\H$; say $\Psi$ is a representative 
vector for $\psi$ and $\tilde{\Psi}$ is a representative vector for 
$\tilde{\psi}$.

Any two unit-vectors can be connected by a unitary operator, so let $V$ be a 
unitary operator such that $V\Psi =  \tilde{\Psi}$. We may write $V$ in the form
$V=e^{iA}$ where $A$ is self-adjoint operator with spectrum ${\rm Sp}(A)\subset 
[-\pi,\pi]$ and hence $\|A\|\leq \pi$. Now  $\cup_{n\in\NN}\M_n'$ is dense 
in $\B(\H)$ in the strong operator topology. Thus by an application of Kaplansky's 
density theorem there exists a sequence of self-adjoints $n\mapsto A_n'\in\M_n'$
strongly converging to $A$ such that $\|A_n\|\leq \pi$ for all $n\in\NN$. Then 
$n\mapsto U_n':=e^{iA_n'}\in\M_n'$ is a sequence of unitary operators strongly 
converging to $V$; in particular $U_n'\Psi\to \tilde{\Psi}$ as $n\to \infty$.

For the von Neumann algebra $\M_n$ the vectors $\Psi$ and $U_n'\Psi$ represent
the same state. Hence as $n\to \infty$,
\begin{equation}
\|\psi_n-\tilde{\psi}_n\| \leq 2 \|U_n'\Psi-\tilde{\Psi}\| \to 0
\end{equation}
which is exactly what we have claimed.
\end{proof}
For what follows, recall our definition of a {\it double-cone} $K^\diamond$
with base $K$. Recall also
that in the beginning of this section 
we have fixed a spacelike hyperplane $H$ and 
some further objects related to $H$. 
\begin{theorem}
Let $(\A,U)$ and $(\tilde{\A},\tilde{U})$ be two algebraic 
QFT models on the $d+1$ dimensional Minkowskian spacetime
satisfying the basic requirements (1,2,3,4,5,6) as well as 
the Bisognano-Wichmann (7) and split (8) properties.
If there exists a unitary
$V$ such that
$$
V\A(K^\diamond)V^* = \tilde{\A}(K^\diamond)
$$
for every double-cone $K^\diamond$ with base $K\subset H^+$,
then the two models are equivalent. That is, there exists a unitary 
operator $W$ such that $W\A({\mathcal O})W^*=\tilde{\A}({\mathcal O})$ for all double-cones 
${\mathcal O}$ and $WU(g)W^*= 
\tilde{U}(g)$ for all elements $g$ of the connected part of the Poincar\'e 
group.
\end{theorem}
\begin{proof}
Let $\Omega$ and $\tilde{\Omega}$ be the (up to phase unique, normalized) 
vacuum vectors for $U$ and $\tilde{U}$, respectively. We may assume that 
the two models are given on the same Hilbert space $\H$ and that $V$ is the 
identity operator so that actually $\A(K^\diamond)=\tilde{\A}(K^\diamond)$ 
for every double-cone $K^\diamond$ with base $K\subset H^+$.  

Remember we defined $\A(\W^+)$ to be the von Neumann algebra 
generated by {\it all} local algebras $\A({\mathcal O})$ with ${\mathcal O}\subset \W^+$. That 
is, theoretically we should take account of {\it all} double-cones 
included in $\W^+$ and 
not only those with bases on $H^+$. However, it is easy to see that one 
can take an increasing sequence of double-cones $n\mapsto K_n^\diamond$ 
with bases on $H^+$ such that not only $\cup_{n\in\NN} K_n^\diamond = 
\W^+$, but actually every bounded region ${\mathcal O}\subset \W^+$ is included in 
$K_n^\diamond$ for {\it some} $n\in\NN$. Then by isotony 
$\A(\W^+)=\{\cup_{n\in\NN}\A(K^\diamond_n)\}''$. So the assumed equality 
implies that $\A$ and $\tilde{\A}$ coincide on $\W^+$, too.

We may assume that $\tilde{\Omega}$ is in the natural cone of $(\Omega, 
\A(\W^+))$.
Indeed, suppose originally it was not so, and consider the state on 
$\A(\W^+)$ given by the vector $\tilde{\Omega}$. This state has a unique 
representative vector $\tilde{\Omega}^\natural$ in the cone in 
question. Since $\tilde{\Omega}$ is cyclic and separating for $\A(\W^+)$, 
the corresponding state is faithful, $\tilde{\Omega}^\natural$ is also 
cyclic and separating for $\A(\W^+)$ and there exists unitary 
$V'\in\A(\W^+)'$ such 
that $V'\Omega = \tilde{\Omega}^\natural$. Then we may replace 
$(\tilde{\A},\tilde{U})$ 
with vacuum vector $\tilde{\Omega}$ by $(V\tilde{\A} V'^*, V\tilde{U} 
V'^*)$ with vacuum vector $\tilde{\Omega}^\natural$. For this latter 
choice we have the desired 
property that its vacuum vector is in the required cone, and 
since $V'$ commutes with all algebras $\A({\mathcal O})\subset \A(\W^+)\; ({\mathcal O}\subset 
\W^+)$, we still have that $\A(K^\diamond) = \tilde{\A}(K^\diamond) = 
V'\tilde{\A}(K^\diamond)V'^*$ 
for every double-cone $K^\diamond$ with base $K\subset H^+$.

Let $\gamma_n$ be the adjoint action of the product 
$U(\tau^n)^*\tilde{U}(\tau^n)$. By what was assumed, we have that for 
every $n\in\NN$
\begin{equation}
\gamma_n(\A(K^\diamond)) = \A(K^\diamond) \; \textrm{ for every double 
cone } K^\diamond \textrm{ with base } K\subset H^+.
\end{equation}
Now let $\omega$ and $\tilde{\omega}$ be the faithful normal states on 
$\A(\W^+)$ given by the vectors $\Omega$ and $\tilde{\Omega}$, 
respectively. 
Then $\omega\circ\gamma_n$ is
nothing else than the state given by the vector 
$\tilde{U}(\tau^n)^*U(\tau^n)\Omega = \tilde{U}(\tau^n)^*\Omega$;
so $\omega\circ\gamma_n = \omega\circ{\rm Ad}(\tilde{U}(\tau^n))$.
On the other hand, $\tilde{\omega}\circ{\rm Ad}(\tilde{U}(\tau^n))= 
\tilde{\omega}$ since
$\tilde{\Omega}$ is an invariant vector for $\tilde{U}(\tau)$. Putting it 
together, and applying our previous lemma we have that 
\begin{equation}
\|\omega\circ\gamma_n - \tilde{\omega} \| =
\|(\omega - \tilde{\omega})\circ{\rm Ad}(\tilde{U}(\tau^n))\| =
\|({\omega}|_{\A(\tau^n(\W^+))} - {\tilde{\omega}}|_{\A(\tau^n(\W^+))}\| 
\to 0
\end{equation}
as $n\to \infty$, since on $\tau^n(\W^+)\subset \W^+$ (by a similar 
argument than that used for $\W^+$) the nets $\A$ and $\tilde{\A}$
coincide and hence
$\tilde{U}(\tau^n)\A(\W^+)\tilde{U}(\tau^n)^* = \A(\tau^n(\W^+))$ for 
every $n\in\NN$.

Since $\Omega$ is cyclic and separating for $\A(\W^+)$, we can find a 
unitary $W_n$ implementing $\gamma_n$ on $\A(\W^+)$ such that 
$W_n\Omega$ is in the natural cone of $(\Omega, \A(\W^+))$. Then 
$W_n\Omega$ and $\tilde{\Omega}$ are exactly the vector representatives in 
the specified natural cone of the states $\omega\circ\gamma_n$ and 
$\tilde{\omega}$, respectively. Thus by the established norm convergence 
of states we have that $\|W_n\Omega-\tilde{\Omega}\|\to 0$ as $n\to 
\infty$.

Let $K^\diamond$ be a nonempty double-cone with base $K\subset 
\tau(H^+)\subset H^+$. Then
the {\it split property} together {\it isotony} 
and {\it Reeh-Schlieder property} imply that   
$\A(K^\diamond)\subset \A(\W^+)$ is a split 
inclusion for which the vacuum vectors 
$\Omega,\tilde{\Omega}$ are standard vectors. 
Moreover, as was discussed in section \ref{sec:prel}, $\A(\W^+)$
is a factor.
Hence by Corollary \ref{main_cor} there
exists unitary operator $W$ and a subsequence $s$ such that 
$n\mapsto W_{s(n)}$ converges strongly to $W$.

It is evident that for the limit $W$ we still have that 
$W\A(K^\diamond)W^*=\A(K^\diamond)=\tilde{A}(K^\diamond)$ for every 
double-cone $K^\diamond$ with base $K\subset H^+$ (and so also for 
regions like $\W^+$ and $\tau^n(\W^+)$), but 
now we also have that $W\Omega=\tilde{\Omega}$. By the
{\it Bisognano-Wichmann property} it immediately follows that 
$WUW^*$ and $\tilde{U}$ coincide on both the boosts associated 
to $\W^+$ and to $\tau(\W^+)$.

Now a quick check shows that the subgroup generated by such boosts 
contains $\tau$ so actually we also have that 
$WU(\tau)W^*=\tilde{U}(\tau)$. Since every double-cone with base on 
$H$ can be shifted into $H^+$ by a repeated use of $\tau$, this further 
implies that $W\A(K^\diamond)W^* = \tilde{\A}(K^\diamond)$ for every 
double-cone $K^\diamond$ with base $K\subset H$ and hence also for 
infinite regions like wedges whose ``edges'' are included in $H$. Then in 
turn --- again by the 
{\it Bisognano-Wichmann property} --- we have that $WUW^*$ and 
$\tilde{U}$ coincides on every boost that is associated to {\it some}
wedge with edge in $H$. But elementary geometric arguments show that such 
boosts generate the entire connected Poincar\'e group so at this point we 
have that $WUW^*(g) = \tilde{U}(g)$ for every element $g$. Moreover, since 
every double-cone can be moved by a suitable Poincar\'e transformation so 
that its base will be on $H$, we now have that 
$W\A(K^\diamond)W^* = \tilde{\A}(K^\diamond)$ for {\it every} double-cone.
Thus $W$ establishes an equivalence between the two models, which is 
exactly what we wanted to prove.
\end{proof}

\section{The conformal case}
\label{sec:conformal}

The conformal chiral QFT, though originally defined on a lightline, can be 
naturally extended to the compactified lightline which is customely 
identified with the circle $S^1\equiv\{z\in\CC| \, |z|=1\}$. On the circle 
the theory becomes {\it M\"obius covariant}; that is, it will carry a 
symmetry action of the group of diffeomorphisms of $S^{1}$ of the form $z 
\mapsto \frac{az+b}{\overline{b}z+\overline{a}}$, which is called the 
{\bf M\"obius group} and is isomorphic to ${\rm PSL}(2,\RR)$.
The connection between the ``circle picture'' and the ``line picture'' 
(here ``line''$\equiv \RR$) is made
by puncturing the circle at $-1\in S^1$ and using 
a Cayley-transformation:
\begin{equation}
x = i \frac{1+z}{1-z}\in \RR  \;\;\Longleftrightarrow z = 
\frac{x-i}{x+i} \in S^1\setminus\{-1\}.
\end{equation}
Via the line picture one can view
{\it translations} and {\it dilations} as 
diffeomorphisms of $S^1$ and in this sense they are 
elements of the M\"obius group. 

A M\"obius covariant net of von Neumann algebras on $S^1$ is  
a map $\A$ which assigns to every nonempty, nondense open ``arc'' 
(or simply {\it interval}) $I\subset S^1$ a von Neumann algebra $\A(I)$ 
acting on a fixed Hilbert space $\H$, together
with a given strongly continuous representation $U$ of 
the M\"obius group satisfying certain properties.
Here we shall not dwell much neither on the defining properties of a 
M\"obius covariant net of von Neumann algebras on $S^1$, nor on their 
known consequences. We only assert that the defining properties are 
adopted versions of (1,2,3,4,5) whereas (the adopted versions of) property 
(6,7) --- that is, the {\it Reeh-Schlieder} and {\it Bisognano-Wichmann 
properties} --- are consequences. One also has {\it irreducibility, 
factoriality of local algebras} and moreover {\it 
additivity} even for an infinite set of intervals: $\vee_{I_\alpha} \A(I_\alpha) 
\supset \A(I)$ whenever $\cup_\alpha I_\alpha \supset I$ for {\it any} collection 
$\{I_\alpha\}$.
For details we refer to \cite{FrG,BGL,GL96,FrJ,GLW}.
Note however that one cannot derive {\it split property} (i.e.\! that 
$\A(K)\subset \A(I)$ is a split inclusion whenever $\overline{K}\subset I$), 
since by taking infinite tensorial products it is easy to construct 
non-split M\"obius covariant nets. 
Nevertheless, it is known to hold in the 
majority of ``interesting'' model. 

\begin{theorem}
Let $(\A,U)$ and $(\tilde{\A},\tilde{U})$ be M\"obius covariant nets of von 
Neumann algebras on $S^1$ with at least one of them being split. Then any of the 
following $4$ conditions:
\begin{itemize}
\item
$\exists$ a unitary 
$W$ s.t. $W\A(I)W^* = \tilde{\A}(I)$ and $WU(g)W^*= \tilde{U}(g)$ for 
all $I,g$,
\item
$\exists$ a unitary 
$V$ s.t. $V\A(I)V^* = \tilde{\A}(I)$ for all $I$,
\item
$\exists$ an (open, nonempty) $I$ and a unitary $V$ s.t.
$V\A(K)V^* = \tilde{\A}(K)$ for all $K\subset I$,
\item
$\exists$ a unitary $V$ s.t.\!
with $\RR$-picture notations
$V\A(j,k)V^* = \tilde{\A}(j,k)$ for 
all $j,k\in\NN$,
\end{itemize}
implies the remaining three.
 \end{theorem}
\begin{proof}
It is clear that any of the conditions implies that if one of the nets is split 
then so is the other and that each condition implies the next one. 
All we have to show is that the last one implies the first one, which can be done by 
simply copying the argument of the proof of the main theorem of the previous 
section. 

Note that 
by (the infinite version of) {\it additivity} the last condition implies that 
for the unitary $V$ appearing in the condition we also have $V\A(k,\infty)V^* = 
\tilde{\A}(k,\infty)$ for every $k\in\NN$.
So we may replace the wedge $\W^+$ in our former proof by the half-line 
$(0,\infty)$. We have to be careful to 
use a translation $\tau$ by an integer length; say we let $\tau$ to be the unit 
translation $x\mapsto x+ 1$. For a split inclusion we can choose $\A(1,2)\subset 
\A(0,\infty)$. 
Then the argument of the mentioned proof shows that there exists a 
unitary $W$ such that $W\A(j,k)W^* = \tilde{\A}(j,k)$ for all $j,k\in\NN$
and moreover $W\Omega = \tilde{\Omega}$ where $\Omega$ and 
$\tilde{\Omega}$ are vacuum vectors for $U$ and $\tilde{U}$, respectively.

From here on the proof is 
actually even simpler than in the ``normal'' case. Indeed, whereas there the 
respective 
modular unitaries did {\it not} generate the Poincar\'e group and so we 
needed to consider further regions, here we do not need any further 
argument. It is easy 
to see that the ``dilations'' associated to intervals of the form $(j,k)$ 
$(j,k\in\NN)$ generate the entire M\"obius group. Moreover, the M\"obius group acts 
transitively on the set of (open, nondense, nonempty) intervals.
So we immediately have that  
$WU(g)W^*= \tilde{U}(g)$ for all $g$ and $W\A(I)W^* = \tilde{\A}(I)$
for all $I$.
\end{proof}
Let us talk about the possible implications of this result regarding half-sided 
modular inclusions. A net $\A$ on the circle is {\bf strongly additive} iff
$\A(I_1)\vee \A(I_2)= \A(I)$ whenever $I,I_1$ and $I_2$ are 
intervals with the last two being obtained from $I$ by the removal of a point.
As was explained in the introduction, 
there is a one-to-one correspondence between {\it strongly additive} 
M\"obius covariant nets 
and standard half-sided modular inclusions of factors.

For any inclusion of von Neumann algebras $\N\subset \M$ 
with a common cyclic vector $\Omega$ consider the 
tunnel introduced by R.\! Longo:
\begin{equation}
\label{tower}
\N_0\supset \N_1\supset \N_2 \supset \N_3 \supset \ldots 
\end{equation}
where $\N_0=\M, \, \N_1=\N$ and $\N_{k+1} = J_k \N_{k-1}' J_k$ 
$(k=1,2,\ldots)$ and $J_k$ is the modular conjugation associated to 
$(\N_k,\Omega)$. 

It is easy to see that the tunnel is well-defined (i.e.\! that $\Omega$ 
remains cyclic and separating at each step of the induction and hence the 
modular conjugation can be indeed considered). But how does it 
depend on the choice of the common cyclic vector $\Omega$? In some sense 
not much. The following statement is included for reasons of 
self-containment; it is well-known to experts of the field.
\begin{lemma}
Let both $\Omega$ and $\tilde{\Omega}$ be common cyclic vectors for the 
inclusion of von Neumann algebras $\N\subset\M$, and denote by 
$\N_0\supset \N_1 \supset \N_2 \supset \ldots$
and
$\tilde{\N}_0\supset \tilde{\N}_1 \supset \tilde{\N}_2 \supset \ldots$
the respective tunnels
defined after equation (\ref{tower}). Then for each $n\in\NN$ 
there exists a unitary operator $V_n$ such that 
$$
V_n \N_k V_n^* = \tilde{\N}_k \;\;\; \textrm{ for all } 
k\in \{0,1,\ldots, n\}.
$$
That is, up to any finite level, the two tunnels are unitarily equivalent.
\end{lemma}
\begin{proof}
We set $V_1=\mathbbm 1$ and define $V_n$ inductively.
Now for $n=1$ the condition is satisfied since by assumption
$\N_0=\tilde{\N}_0=\M$ and $\N_1=\tilde{\N}_1=\N$. So assume $V_k$
is already defined in a way satisfying the requirement made in the 
statement. Then $V_k\Omega$ is cyclic and separating for $(V_k\N_k 
V_k^*)=\tilde{\N}_k$ so there is a unitary $U_k\in 
\tilde{\N}_k$ such that $U_k V_k\Omega$ is in the
natural cone of $(\tilde{\Omega},\tilde{\N}_k)$. Set $V_{k+1}:=U_k V_k$; 
it is then evident that $V_{k+1}\N_j V_{k+1}^* = \tilde{\N}_j$ for all 
$j=0,1,\ldots ,k$. Moreover, as 
$V_{k+1}\Omega = U_k V_k \Omega$ is in the natural cone of 
$(\tilde{\Omega},\tilde{\N}_k)$ and $V_{k+1}\N_kV_{k+1}^* = \tilde{\N}_k$,
we have that the adjoint action of $V_k$ takes 
the modular conjugation $J_k$ associated to $(\Omega,\N_k)$ 
into the modular conjugation $\tilde{J}_k$
associated to $(\tilde{\Omega},\tilde{\N}_k)$.
Thus
\begin{equation}
V_{k+1}\N_{k+1}V_{k+1}^*= V_{k+1} J_k \N_{k-1}' J_k V_{k+1}^* = 
\tilde{J}_k (V_{k+1}\N_{k-1}V_{k+1}^*)' \tilde{J}_k = \tilde{J}_k \tilde{\N}_{k-1}' 
\tilde{J}_k 
= \tilde{\N}_{k+1}
\end{equation}
and hence the statement is proved by induction.
\end{proof}

Let $(\A,U)$ be a M\"obius covariant net with vacuum vector $\Omega$ and 
denote the modular objects associated to $(\Omega, \A(k,\infty))$ by $J_k$ and 
$\Delta_k$. Using the {\it Bisognano-Wichmann property} and 
the main theorem \cite[Thm.\! 2.1]{AZs} of half-sided modular inclusions, the
product $J_k J_{k-1}$ can be expressed with the modular unitaries which in turn 
can be expressed by $U$ resulting in $J_kJ_{k-1} = U(\tau)^2$ where $\tau$ 
is the unit-translation defined in the $\RR$-picture by the map $x\mapsto 
x+1$. Hence
\begin{equation}
J_k \A(k-1, \infty)' J_k = 
J_k J_{k-1} 
\A(k-1,\infty)
J_{k-1} J_{k} = U(\tau)^2
\A(k-1,\infty) U(\tau)^{-2} = 
\A(k+1,\infty)
\end{equation}
and so the tunnel (\ref{tower}) associated to $(\Omega,\A(0,\infty)\subset 
\A(1,\infty))$ is nothing else than the sequence of inclusions
\begin{equation}
\A(0,\infty)\subset
\A(1,\infty)\subset
\A(2,\infty)\subset \ldots
\end{equation}
Note that in case we have {\it strong additivity}, by taking relative commutants 
this sequence determines all algebras of the form $\A(j,k)$ with $j,k\in\NN$. 
{\it Vice versa}, if we know $\A(j,k)$ for all $j,k\in\NN$ then by (the 
infinite version 
of) {\it additivity} we can compute all algebras of the form 
$\A(k,\infty)$ with 
$k\in\NN$. So by what was explained we can draw the following conclusion.
\begin{corollary}
Suppose $(\Omega,\N\subset \M)$ and $(\tilde{\Omega},\tilde{\N}\subset 
\tilde{\M})$ are two standard half-sided modular inclusions of factors and 
denote the two corresponding strongly additive M\"obius covariant nets 
by $(\A,U)$ and $(\tilde{\A},\tilde{U})$, respecively. Then the conditions:
\begin{itemize}
\item $\exists$ a unitary $V$ s.t. $V\M V^* = \tilde{\M},\, 
V\N V^* = \tilde{\N}$,
\item $\forall n\in\NN: \exists $ unitary $V_n$ s.t.
$ V_n\A(j,k)V_n^* = \tilde{\A}(j,k)$ for all $j,k\in\{0,1,\ldots,n\}$,
\end{itemize}
are equivalent.
\end{corollary}
The relevance of this statement 
in light of the conformal version of our algebraic Haag's theorem
has been already discussed in the introduction.
\vspace{0.8cm}

\noindent
{\bf \large Acknowledgment.} The author would like to thank 
Roberto Longo, Sebastiano Carpi and Yoh Tanimoto for useful 
discussions.

\end{document}